\documentclass[12pt]{article}
\usepackage[final]{epsfig}
\usepackage{natbib}
\setcitestyle{numbers,square}
\usepackage{float}
\usepackage[T1]{fontenc}
\usepackage{authblk}
\usepackage{natbib}
\usepackage{latexsym}
\usepackage{comment}
\usepackage{amsthm}
\usepackage{verbatim,color}
\usepackage{graphicx}
\usepackage{amssymb}
\usepackage{amssymb, color}
\usepackage{amsmath}
\usepackage{latexsym}
\usepackage{wrapfig}
\usepackage{subfig}
\usepackage[utf8]{inputenc}
\usepackage{multicol}
\usepackage{url}
\usepackage{siunitx}
\usepackage{hyperref}
\usepackage{doi}
\usepackage{makecell}
\usepackage{enumitem}

\makeatletter
\def\blfootnote{\xdef\@thefnmark{}\@footnotetext}
\makeatother

\newcommand{\bfa}{{\bf a}}

\newcommand{\bfc}{{\bf c}}

\newcommand{\bfe}{{\bf e}}
\newcommand{\bff}{{\bf f}}

\newcommand{\bfm}{{\bf m}}

\newcommand{\bfr}{{\bf r}}
\newcommand{\bfs}{{\bf s}}
\newcommand{\bft}{{\bf t}}

\newcommand{\bfx}{{\bf x}}

\newcommand{\bfz}{{\bf z}}

\newcommand{\bfI}{{\bf I}}

\newcommand{\bfQ}{{\bf Q}}
\newcommand{\bfR}{{\bf R}}

\newcommand{\beq}{\begin{equation}}
\newcommand{\eeq}{\end{equation}}
\newcommand{\beqs}{\begin{eqnarray}}
\newcommand{\eeqs}{\end{eqnarray}}

\newcommand{\calG}{{\cal G}}
\newcommand{\calH}{{\cal H}}

\newcommand{\calR}{{\cal R}}
\newcommand{\calS}{{\cal S}}
\newcommand{\calT}{{\cal T}}

\newcommand{\Rmnum}[1]{\text{\uppercase\expandafter{\romannumeral #1}}}
\newtheorem{theorem}{Theorem}[section]

\newtheorem{definition}{Definition}[section]

\newcommand{\bfig}{\begin{figure}[!h]}	
\newcommand{\efig}{\end{figure}}		
	
\newcommand{\T}{\mathrm{T}}

\begin{document}

\begin{center}
{\huge Objective Moiré Patterns} \\

\vspace{5mm}
\normalsize

\vspace{2mm}
{\large
Fan Feng}

\vspace{2mm}
Email: fanfeng@pku.edu.cn

\vspace{2mm}
{Department of Mechanics and Engineering Science, College of Engineering,\\
Peking University, Beijing 100871,  China}
\end{center}

\vspace{0.2in}
\normalsize

\noindent {\normalsize {\bf Abstract.}  } 
Moiré patterns, typically formed by overlaying two layers of two-dimensional materials, exhibit an effective long-range periodicity that depends on the short-range periodicity of each layer and their spatial misalignment. Here, we study moiré patterns in objective structures with symmetries different from those in conventional patterns such as twisted bilayer graphene. Specifically, the mathematical descriptions for ring patterns, 2D Bravais lattice patterns, and helical patterns are derived analytically as representative examples of objective moiré patterns, using an augmented Fourier approach.
Our findings reveal that the objective moiré patterns retain the symmetries of their original structures but with different parameters. In addition, we present a non-objective case, conformal moiré patterns, to demonstrate the versatility of this approach. We hope this geometric framework will provide insights for solving more complex moiré patterns and facilitate the application of moiré patterns in X-ray diffractions, wave manipulations, molecular dynamics, and  other fields. 
\tableofcontents

\section{Introduction}
Moiré patterns \cite{oster1963moire} have gained significant attention since the discovery of unconventional superconductivity in twisted bilayer graphene (TBG) \cite{cao2018}. 
The term moiré originates from the French word for “watered silk”, referring to two layers of fabric pressed together to produce interference patterns.
In the scientific community, these patterns are typically formed by superposing two layers of 2D materials with a small misalignment. For example, the moiré pattern of TBG is constructed by stacking two layers of graphene that differ by a
`magic angle' $\sim 1.1^\circ$.
Moiré patterns have been extensively studied in many fields, including the mechanical behavior of stacked bilayer graphenes \cite{zhang2018structural,dai2020mechanics}, inverse design of moiré images \cite{design_moire}, superconductivity \cite{yankowitz2019tuning}, etc. Among these studies, the mathematical description of the geometry of moiré patterns plays a foundational role. While 2D moiré patterns with translational symmetry have been thoroughly explored, mathematical descriptions for patterns with other types of symmetry remain limited.  
In this paper, we focus on the mathematical formulation for moiré patterns within the context of objective structures possessing different symmetries. 

\begin{figure}[!h]
    \centering \includegraphics[width=1.0\textwidth]{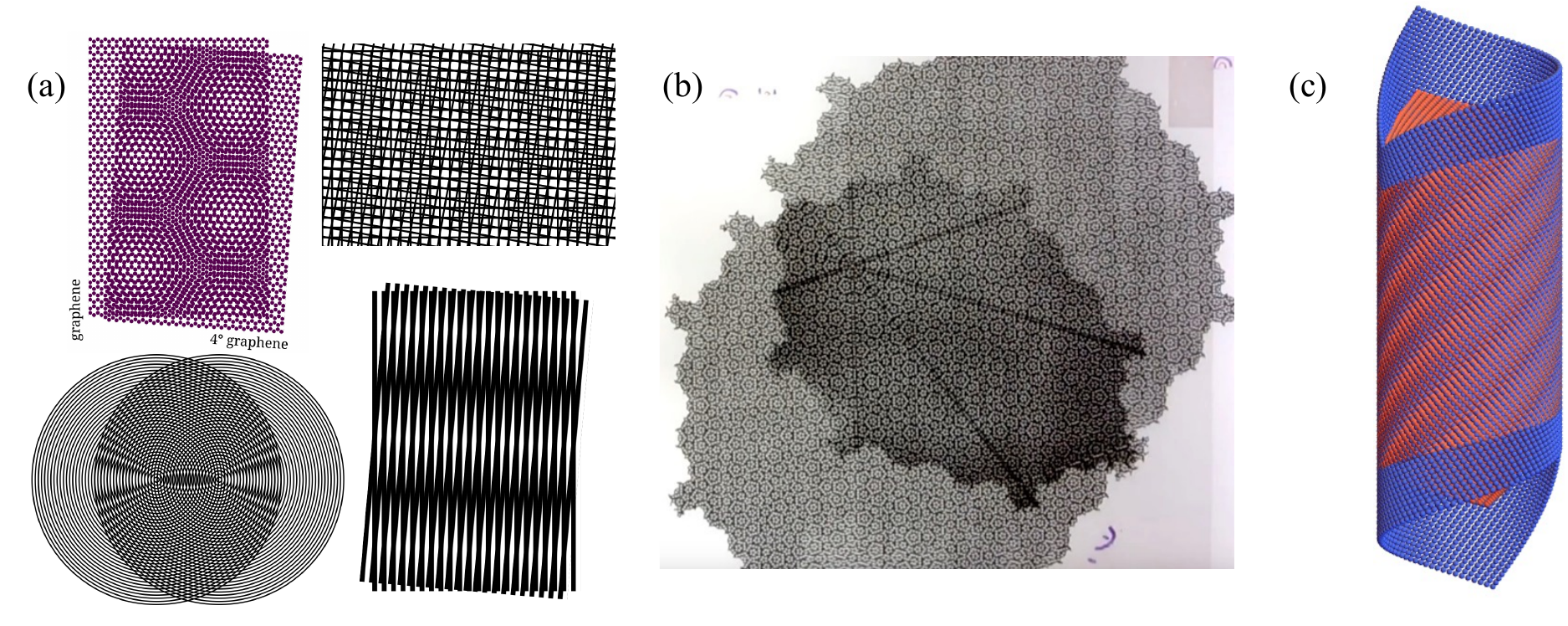}
    \caption{(a) From upper left to lower right: moiré patterns created by twisting bilayer graphenes, superposing two square grids, superposing two sets of concentric circles, and superposing two sets of parallel lines. (b) Lines emerge in the superposition of two Penrose tilings. Picture adapted from \cite{penrose}. (c) Helical moiré pattern in this paper.}
    \label{fig:intro}
\end{figure}
Conventional moiré patterns are created by overlaying two 2D lattice structures on a plane. These 2D structures can include graphenes, Bravais lattices, circular lattices, or simple sets of lines/curves (Fig.~\ref{fig:intro}(a)). 
Beyond 2D symmetric lattices, however, many lattices, such as quasi-crystals, have no translational or rotational symmetry. A well-known example is the aperiodic Penrose tiling discovered by
Sir Roger Penrose around 1975. Recently, he showed in his talk that carefully superposing two identical Penrose tilings produces `Incredible moiré Patterns' such as lines (Fig.~\ref{fig:intro}(b)) and nets \cite{penrose}. 
Mathematically, there are mainly two approaches to computing the geometry of moiré patterns: (1) counting the intersections of curves when the original patterns consist of lines/curves \cite{oster1964theoretical,rogers1977geometrical}, and (2) superposing two intensity functions maximized at corresponding points (the Fourier approach) when the original patterns consist of lattice points \cite{design_moire}. In this paper, we will still focus on the symmetric case but extend it by exploring moiré patterns with different symmetries in 3D. Here we use the Fourier approach to compute the geometry of moiré patterns for objective structures possessing various symmetries.

The concept of objective structures, first introduced by Richard James in 2006 \cite{james2006objective}, is a general way of describing lattice structures with various symmetries. In his heuristic definition, \emph{an objective atomic structure is a collection of atoms in which each atom sees the same environment}, up to orthogonal transformations (rotations and mirror reflections). Thus, many structures in nature, such as graphenes, C60, tails of bacteriophage T4 and black phosphorus, are objective structures. A fundamental theorem in this framework is that every objective structure (OS) can be generated by taking the orbit of a discrete isometry group on a finite set of lattice points in 3D \cite{james2006objective}. Thus this theorem allows us to study OS by analyzing the corresponding isometry groups, harnessing the multiplication, inverse and other rules of groups.
On the computational side, a key development within the OS framework is objective molecular dynamics (OMD), introduced by Dumitrică and James \cite{dumitricua2007objective}. By leveraging group operations and symmetry, OMD significantly reduces the dimensionality of Newtonian dynamics compared to conventional molecular dynamics (MD), making it more efficient for simulating graphene bending \cite{nikiforov2014tight, nikiforov2010wavelike}, dislocations \cite{zhang2009dislocation,pahlani2023objective0}, and fluid mechanics \cite{dayal2012design,pahlani2023objective,pahlani2023constitutive}.
Additionally, Banerjee extended the OS framework to density functional theory (DFT) \cite{banerjee2015spectral,banerjee2021ab}, which has been applied to first-principles simulations of bending in nanostructures \cite{banerjee2016cyclic}. In multiscale modeling, the OS framework has led to the development of a quasicontinuum approach \cite{hakobyan2012objective} and a multiscale hierarchy framework \cite{james2024kinetic}.
On the mathematical side, the OS framework has been applied in studying the interaction between partial differential equations (PDEs), group operations, and objective structures. For example, the coupling of Maxwell’s equations with helical structures has led to the discovery of helical Bragg diffraction, which enables precise measurements of lattice parameters \cite{friesecke2016twisted,justel2016bragg}. Furthermore, the Boltzmann equation, when incorporating translational symmetry, exhibits a reduced dimensionality, facilitating the study of long-time asymptotic nonequilibrium behaviors \cite{james2019long,james2020long,james2019self}.
On the physical and engineering side, the OS framework has been applied to helical phase transformations \cite{feng2019phase,ganor2016zig}, the design of helical origami \cite{feng2020helical}, origami structures with curved tiles \cite{liu2024design}, artificial blood vessels \cite{velvaluri2021origami}, and other engineering applications.
Moreover, the concept of OS has been further extended to “locally objective structures”, in which atoms experience the same local environment within a finite neighborhood \cite{rieger2022locally}.

In this paper, we explore another theme—moiré patterns—within the framework of objective structures. Inspired by classic planar moiré patterns in 2D Bravais lattices, 
we employ an augmented Fourier approach to calculate the geometry of objective moiré patterns. The basic idea is to construct two wave functions with peaks located at lattice points, such that the moiré pattern emerges as the `wave packet' of the two superposed waves. Three examples (ring structures, 2D Bravais lattice structures and helical structures) of objective moiré patterns are theoretically explored, and the analytical geometric structures of moiré patterns are provided.
We notice that the  moiré patterns and the effective 1D heterostructures for multi-layer carbon nanotubes have been studied \cite{zhao2022interlayer,koshino2015incommensurate}, which is a special case of our general helical patterns.
Additionally, in our paper a cone-like conformal moiré pattern (which is not objective) generated by conformal groups is computed analytically to indicate the versatility of the group orbit approach. Our results show that the objective moiré patterns and the conformal patterns retain the same symmetries of the original structures but with different parameters.
This is not the end of the story. Notably, all the examples presented here can be mapped isometrically onto a plane, making the wave functions straightforward to construct.
We can thus solve the moiré pattern in 2D, and then map it back to the objective structure framework. However, calculating moiré patterns for general objective structures that cannot be mapped to a plane, such as Buckyballs, remains unclear. The main challenge here is to define a generalized wave function respecting the symmetry of OS. Although helical waves have been found and applied in helical structures \cite{banerjee2015spectral, justel2016bragg}, the problem for general cases is yet to be solved.

This paper is organized as follows. In Section 2, we briefly introduce the concepts of objective structures and isometry groups. Then we will discuss the mathematical description of the objective moiré pattern created by rings, 2D lattices, and helical structures in Section 3. In Section 4, we showcase an example of conformal moiré pattern. Finally, Section 5  concludes the main points of the paper and envisions the possibility of solving the general problem.

\section{Preliminaries on objective structures}

The most common periodic atomic structures are Bravais lattices, which are constructed by translating three linearly independent primitive lattice vectors in 3D, i.e. $\{n_1\bfa_1 + n_2 \bfa_2 + n_3 \bfa_3: n_i \in \mathbb{Z}\}$. Based on their symmetry groups, Bravais lattices can be classified into 14 lattice systems. However, not all periodic structures are translational invariant. For example, structures such as Fullerenes, the tail of bacteriophage T4 and the carbon nanotubes indeed exhibit symmetries distinct from those in Bravais lattices. Objective structures serve as a generalized concept that encompasses periodic atomic structures with a variety of symmetries. 

As first proposed by James, an objective (atomic) structure is a group of atoms $\calS = \{\bfx_1,\bfx_2,\dots, \bfx_N\}$ in which `each atom sees the same environment' up to orthogonal transformations (rotations and mirror reflections) \cite{james2006objective}. For instance, each atom in Bravais lattices sees the same environment due to translational symmetry, making Bravais lattices examples of objective structures. This intuitive interpretation also applies to other types of periodic structures. A more precise mathematical definition is provided below.
\begin{definition}
    $\calS = \{\bfx_1,\bfx_2,\dots, \bfx_N\}$ is an objective atomic structure if there are orthogonal transformations $\{\bfQ_1, \bfQ_2,\dots,\bfQ_N\}$ such that 
    $\calS = \{\bfx_i + \bfQ_i (\bfx_k - \bfx_1): k=1,2,\dots,N\}$.
\end{definition}
In this definition, $\bfQ_i \in O(3)$ is an orthogonal transformation satisfying $\bfQ \bfQ^{\text{T}} = \bfI$ and $N$ can be infinity. Physically, $\bfQ_i$ can be treated as the reorientation of the observer at $\bfx_i$ to achieve the same environment (i.e., the vectors $\bfx_k - \bfx_1$) as $\bfx_1$. It is natural to generalize the definition of objective atomic structures to {\it objective molecular structures} by replacing $\bfx_i$ with a group of atoms $\bfx_{i,j},j=1,2,\dots,M$. The symmetries of OS can also be generalized. Inspired by the Bravais lattice and its corresponding translational group, the symmetry of OS can be described by {\it isometry groups} defined as follows.

\begin{definition}
An isometry $g=(\bfQ|\bfc), \bfQ\in O(3), \bfc \in \mathbb{R}^3$ is an operator that maps a reference object to another and preserves the local distance between points. An isometry group $\calG=\{g_1, g_2, g_3 \dots \}$ is a collection of isometries satisfying the group axioms: closeness, associativity, existence of identity element and inverse element. The orbit of $\calG$ on $\bfx\in\mathbb{R}^3$ is defined as $\calG(\bfx) = \{g_1(\bfx), g_2(\bfx), g_3(\bfx) \dots \}$. $\calG$ is said to be a discrete isometry group if $\calG(\bfx)$ has no accumulating points for any $\bfx$ (omitting repeated atoms).
\end{definition}
In this definition, the action of an isometry on a point $\bfx \in \mathbb{R}^3$ is defined by $g(\bfx) = \bfQ \bfx + \bfc$. Under this action, one may naturally have the multiplication, identity and inverse for the isometry $g_i=(\bfQ_i|\bfc_i)$ as:
\beqs
&&g_1 g_2(\bfx) := g_1(g_2(\bfx)) \implies g_1 g_2 = (\bfQ_1 \bfQ_2 | \bfQ_1 \bfc_2 + \bfc_1), \nonumber \\
&&id = (\bfI|{\bf 0}), \nonumber \\
&& g^{-1} = (\bfQ^{\mathrm{T}}|-\bfQ^{\mathrm{T}} \bfc).
\eeqs
The associativity holds under the multiplication rule, and the closeness is a prior assumption. Since $\bfQ \in O(3)$ (i.e., $\bfQ \bfQ^{\T} = \bfI$), local distances are preserved by $|\bfx_1 - \bfx_2| = |g(\bfx_1) - g(\bfx_2)|=|\bfQ\bfx_1 - \bfQ\bfx_2|$. Without providing the proof here (see \cite{james2006objective}), a key theorem relating objective structures and their corresponding isometry groups is as follows.
\begin{theorem}
If one allows for intersecting images, every objective structure is the
 orbit of a discrete group of isometries on a finite set of points in $\mathbb{R}^3$.
\end{theorem}

This theorem unveils the connection between symmetric structures and isometry groups, suggesting that we may generate objective structures by acting isometry groups on atoms in $\mathbb{R}^3$ and conveniently study the physical properties by examining the corresponding isometry groups. For instance, Bravis lattices are generated by translational groups $\{t_1^p t_2^q t_3^r: t_i=(\bfI|\bfe_i), p,q,r\in \mathbb{Z}\}$.
The main subject of the International Tables of Crystallography corresponds to isometry groups containing three linearly independent translations. A classification of isometry groups, which includes ring groups, helical groups \cite{feng2019phase} and others, can be found in a note \cite{isometry_group}. 

\section{Objective moiré pattern}
In this section, we will use the framework of OS and isometry groups to study the geometry of moiré patterns constructed by overlaying two layers of objective structures, including ring structures, 2D Bravais lattice structures, and helical structures. An approach for more general objective moiré patterns will also be discussed.

\subsection{Ring moiré pattern}
Let $\{\bfe_1,\bfe_2,\bfe_3\}$ be an orthonormal basis and $\bfz$ be a reference point in 3D. A ring lattice structure is a circular objective structure generated by a discrete ring group,
\beq
\calR = \{h^p: p \in \mathbb{Z}\},\quad h=(\bfR_{\theta}|(\bfI - \bfR_{\theta
})\bfz),
\eeq
and the corresponding ring lattice is given by
\beq
\calR({\bf x}) = \{h^p({\bf x}): p \in \mathbb{Z}\}
\eeq
with radius $r=|\bfx-\bfz|$ and $h^p({\bf x})=\bfR_{p\theta}(\bfx-\bfz)+\bfz$.
Here $\bfR_\theta$ is a rotation tensor about $\bfe_3$ of rotation angle $\theta = 2\pi/n$ with $n\in \mathbb{Z}$, $\bfz \cdot \bfe_3 = 0$. We then overlay two rings $\calR_a(\bfx) = \{h_a^p(\bfx): p\in \mathbb{Z} \}$ and $\calR_b(\bfx) = \{h_b^p(\bfx) : p \in \mathbb{Z}\}$  with slightly different generators $h_a$  and $h_b$ to form a ring-like moiré pattern as shown in Fig.~\ref{fig:ring}.

\begin{figure}[!h]
    \centering
    \includegraphics[width=\textwidth]{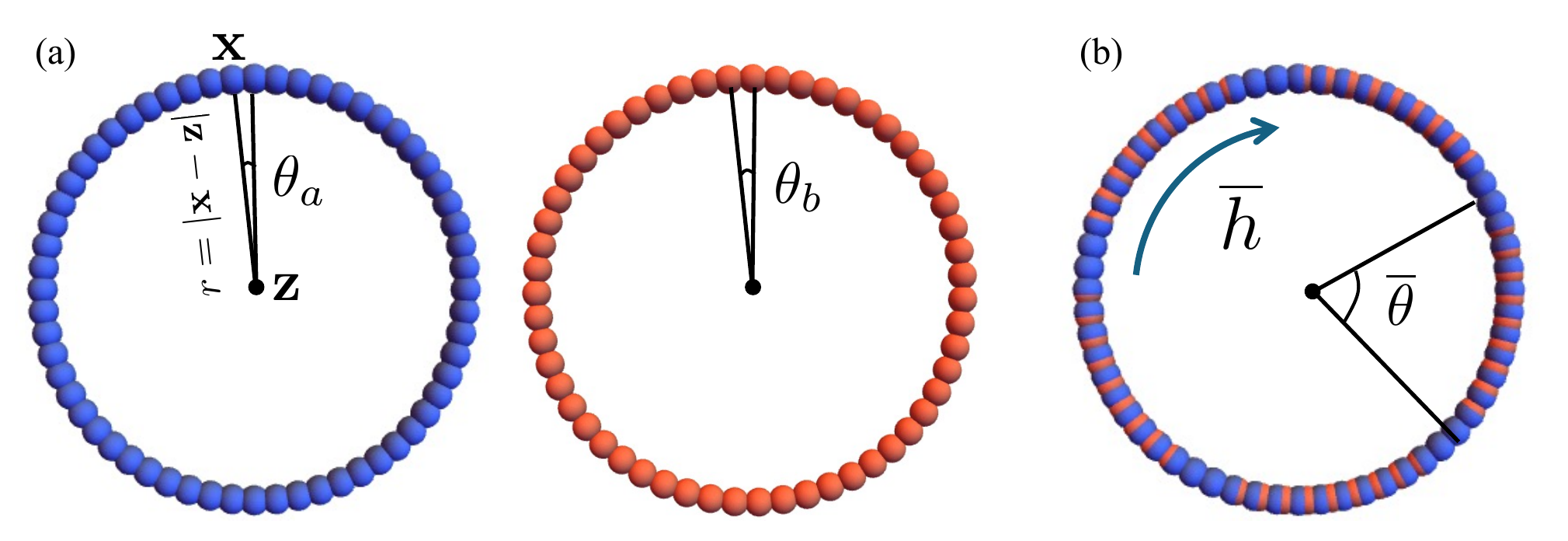}
    \caption{(a) Two ring structures with parameters $r=1, \theta_a=2\pi/60, \theta_b=2\pi/55$. (b) The formed moiré pattern with parameters $r=1, \bar{\theta}=2\pi/5$.}
    \label{fig:ring}
\end{figure}

To compute the geometry of the moiré pattern, it is convenient to map the rings to lines isometrically, and then study the moiré pattern on a line. To this end, the isometric map is simply chosen as
\beq
d(\calR(\bfx)) = |\bfx - \bfz| p \theta \bfe_1,
\eeq
where the 1D lattice vector is $|\bfx - \bfz| \theta \bfe_1$.
We then use the Fourier approach to analyze the moiré pattern by superposing two rings $\calR_a(\bfx)$ and $\calR_b(\bfx)$ with lattice vectors $\bfr_a = |\bfx-\bfz| \theta_a \bfe_1$ and $\bfr_b = |\bfx-\bfz| \theta_b \bfe_1$. To satisfy the periodicity and discreteness conditions, the parameters $\theta_a$ and $\theta_b$ are chosen as
$\theta_a = 2\pi/n_a$ and $\theta_b = 2\pi/n_b$, $n_{a,b} \in \mathbb{Z}\setminus\{0\}$. The corresponding reciprocal lattice vectors are $\hat{\bfr}_{a,b} = \bfe_1/(|\bfx - \bfz|\theta_{a,b})$, respectively. Here we use the notation $\hat{\bfr}$ to represent the reciprocal vector of $\bfr$, including in the 2D case in the following section. It is straightforward to follow the typical Fourier method to compute the moiré pattern by superposing the two wave functions ($i$ is the imaginary unit) for intensity
\beq
f_a(\bfr) = \sum_{j\in\mathbb{Z}} c_j \exp [2\pi i(j \hat{\bfr}_a \cdot \bfr)], \quad f_b(\bfr) = \sum_{j\in\mathbb{Z}} d_j \exp [2\pi i(j \hat{\bfr}_b \cdot \bfr)], 
\eeq
where the intensity reaches maximum at lattice points in each ring.
Then the intensity for the moiré pattern is the superposition
\beq
f_m(\bfr) = f_a (\bfr) + f_b(\bfr)  = \sum_{j\in\mathbb{Z}} c_j \exp[2\pi i (j \hat{\bfr}_a \cdot \bfr)] \left(1+ \frac{d_j}{c_j}\exp[2\pi i (j(\hat{\bfr}_b - \hat{\bfr}_a)\cdot \bfr)] \right),
\eeq
implying the periodicity of the moiré pattern is determined by the reciprocal lattice of $(\hat{\bfr}_b - \hat{\bfr}_a)$. By transforming the linear moiré pattern back to the ring pattern, a direct calculation gives the geometric structure for the moiré pattern as
\beq
\calR_m({\bf x}) = \{\bar{h}^p({\bf x}): p \in \mathbb{Z}\},\quad \bar{h}=(\bfR_{\bar{\theta}}|(\bfI - \bfR_{\bar{\theta}
})\bfz)  ~\text{with}~ \bar{\theta} = 2\pi/(n_a - n_b).
\eeq

 Figure \ref{fig:ring} shows an example of a ring-like moiré pattern constructed by superposing two ring structures with parameters $r=1, \theta_a=2\pi/60, \theta_b=2\pi/55$. Following the calculation of the moiré pattern, the moiré pattern is also a discrete ring, with periodicity $\bar{n} = n_a - n_b = 5$ verified in Fig.~\ref{fig:ring}(b). As a generalization, we prove a theorem about the discreteness and periodicity for helical structures in Section \ref{sec:helical}.

\subsection{2D Bravais moiré pattern}\label{sec:2D}
In this section, we construct 2D Bravais lattices with translation symmetry (which is a special case of objective structures) using isometry groups, and derive the moiré pattern using the Fourier approach by superposing two waves in 2D. The objective description of the moiré pattern will also be discussed.

\subsubsection{2D translational group and moiré pattern}
Within the framework of objective structures, a 2D Bravais lattice with translational symmetry is generated by a translation group.
Let $\bfe_1 = (1,0)$ and $\bfe_2 = (0,1)$ be the orthonormal basis of $\mathbb{R}^2$. A translation group is given by 
\beq
\calT = \{g_1^p g_2^q: (p,q)\in \mathbb{Z}^2 \},\quad g_i = (\bfI | \bft_i),
\eeq
where $g_1$ and $g_2$ are the group generators with $\bft_i = \xi_i \bfe_1 + \eta_i \bfe_2 \in \mathbb{R}^2$. To avoid the degenerate case, the translation parts $\bft_1$ and $\bft_2$ need to be linearly independent, i.e., $\bft_1 \cdot \bft_2^{\perp} \neq 0$.
It is also worth noting that the group $\calT$ is Abelian. To construct a 2D lattice, we apply the group $\calT$ on a point (without loss of generality, we choose ${\bf 0}$) in $\mathbb{R}^2$, namely, $\calT({\bf 0})=\{p \bft_1 + q \bft_2:(p,q)\in \mathbb{Z}^2\}$. 
The procedure introduced above is a typical example of describing an objective structure, equivalent to a 2D Bravais lattice with primitive vectors $\bft_1$ and $\bft_2$. 
Beyond Bravais lattices, the objective structure framework provides a more convenient and general approach for describing structures with complex symmetries.

Now we turn to the 2D moiré pattern. The collection of the lattice points for a 2D Bravais lattice is given by
\beq
\calT_a({\bf 0}) = \{p \bfr_1 + q \bfr_2 : (p,q) \in \mathbb{Z}^2, \bfr_1 \cdot \bfr_2^{\perp} \neq 0 \},
\eeq
where $\bfr_1$ and $\bfr_2$ are the primitive lattice vectors given by some $\xi \bfe_1 + \eta \bfe_2$.  Similarly, the collection of lattice points for the second Bravais lattice  is given by
\beq
\calT_b({\bf 0})  = \{ p \bfs_1 + q \bfs_2: (p,q) \in \mathbb{Z}^2, \bfs_1 \cdot \bfs_2^{\perp} \neq 0 \},
\eeq
where $\bfs_1$ and $\bfs_2$ are the primitive lattice vectors for the second layer. Let the corresponding  reciprocal lattice vectors be $\hat{\bfr}_1, \hat{\bfr}_2, \hat{\bfs}_1, \hat{\bfs}_2$. By definition, we have
\beq
\bfr_i \cdot \hat{\bfr}_j = \delta_{ij},~ \bfs_i \cdot \hat{\bfs}_j = \delta_{ij},~ i,j \in \{1, 2\}.
\eeq
 The bilayer structure exhibits a long-range periodicity, owing to the coupling between the generically distinct periodicities of the two monolayers. To describe the periodicity and the coupling, we impose two general wave functions for intensity
\beqs
f_a(\bfr) &=& \sum_{j,k\in\mathbb{Z}} c_{j,k}\exp[2\pi i (j\hat{\bfs}_1 + k \hat{\bfs}_2) \cdot \bfr], \nonumber \\
f_b(\bfr) &=& \sum_{j,k \in\mathbb{Z}} d_{j,k} \exp [2\pi i (j \hat{\bfr}_1 + k \hat{\bfr}_2) \cdot \bfr], \label{eq:functions}
\eeqs
where $i$ is the imaginary unit and the function magnitude reaches maximum at lattice points in 2D. Let $f_m(\bfr) = f_a(\bfr) + f_b(\bfr)$ be the superposition of these two functions. Substituting Eq.~(\ref{eq:functions}) yields 
\beq
f_m(\bfr) = \sum_{j,k \in\mathbb{Z}} c_{j,k} \exp[2\pi i (j \hat{\bfs}_1 + k \hat{\bfs}_2) \cdot \bfr] \left(1+ \frac{d_{k,j}}{c_{k,j}} \exp[2\pi i (j (\hat{\bfr}_1 - \hat{\bfs}_1) + k (\hat{\bfr}_2 - \hat{\bfs}_2)) \cdot \bfr ]\right). \label{eq:superposition}
\eeq
The small difference in lattice vectors ($(\bfs_1, \bfs_2) = (\bfr_1, \bfr_2)$) yields small difference in the corresponding reciprocal vectors ($(\hat{\bfs}_1, \hat{\bfs}_2) \approx (\hat{\bfr}_1, \hat{\bfr}_2)$). As shown in Eq.~\ref{eq:superposition}, the periodicity of $f_m(\bfr)$ is determined by two terms: the short-range intrinsic periodicity of the substrate determined by $(\hat{\bfs}_1, \hat{\bfs}_2)$ and the modulated long-range periodicity determined by $(\hat{\bfr}_1 - \hat{\bfs}_1, \hat{\bfr}_2 - \hat{\bfs}_2)$. The second term is referred to as {\it moiré vector}, which are given by the reciprocal vectors $(\bfm_1,\bfm_2)=(\widehat{\hat{\bfr}_1-\hat{\bfs}_1},\widehat{\hat{\bfr}_2 - \hat{\bfs}_2})$. The moiré pattern of twisted bilayer graphenes is a case of this type.

\subsubsection{Example}
Figure \ref{fig:2D} shows an example of 2D moiré pattern constructed by two layers of 2D lattices with primitive lattice vectors $\bfr_1=(0.1,0), \bfr_2=(0.0078,0.0097)$ and $\bfs_1=(0.09,0), \bfs_2=(0,0.11)$.  Following the above Fourier approach, we provide the lattice vectors for the moiré pattern as
\beqs
&&\bfm_1 =\widehat{\hat{\bfr}_1-\hat{\bfs}_1}=(0.9,0), \nonumber \\
&&\bfm_2= \widehat{\hat{\bfr}_2 - \hat{\bfs}_2} = (-0.754,1.064),
\eeqs
where the hat $\hat{\bff}_i$ denotes the reciprocal vector corresponding to the basis $(\bff_1,\bff_2)$ in 2D. 

\begin{figure}[!h]
    \centering
    \includegraphics[width=\textwidth]{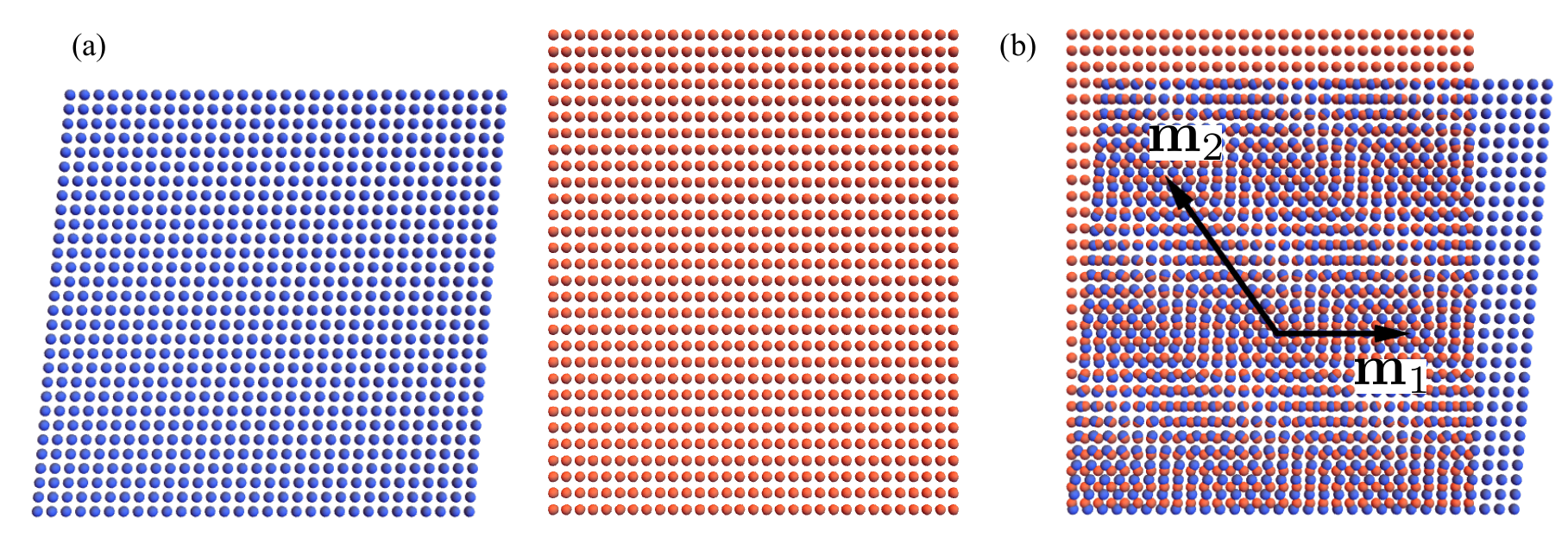
    }
    \caption{(a) Two 2D Bravais lattice structures with lattice vectors $\bfr_1=(0.1,0), \bfr_2=(0.0078,0.0097)$ and $\bfs_1=(0.09,0), \bfs_2=(0,0.11)$. (b) The formed moiré pattern is also a 2D lattice, with lattice vectors $\bfm_1=(0.9,0), \bfm_2=(-0.754,1.064)$.}
    \label{fig:2D}
\end{figure}

\subsection{Helical moiré pattern} \label{sec:helical}
\subsubsection{Helical structures}
A helical structure $\calH(\bfx)$ is generated by applying a helical group $\calH=\left\{g_{1}^{p} g_{2}^{q} :(p, q) \in \mathbb{Z}^{2}\right\}$ on a point $\mathbf{x} \in \mathbb{R}^{3}$, i.e., $\calH(\bfx)=\left\{g_{1}^{p} g_{2}^{q}(\bfx) :(p, q) \in \mathbb{Z}^{2}\right\}$. By reparameterizing helical groups appropriately \cite{feng2019phase}, the group generators $g_{1}$ and $g_{2}$ have the following forms:
\begin{align}
& g_{1}=\left(\mathbf{R}_{\theta_{1}} \mid\left(\mathbf{I}-\mathbf{R}_{\theta_{1}}\right) \mathbf{z}+\tau_{1} \mathbf{e}\right) \nonumber\\
& g_{2}=\left(\mathbf{R}_{\theta_{2}} \mid\left(\mathbf{I}-\mathbf{R}_{\theta_{2}}\right) \mathbf{z}+\tau_{2} \mathbf{e}\right) 
\end{align}
where $\mathbf{R}_{(.)} \in \mathrm{SO}(3), \mathbf{z}, \mathbf{e} \in \mathbb{R}^{3}, \theta_{i}, \tau_{i} \in \mathbb{R}$. 
Here $\bfz$ is the reference point and $\bfe$ is the axis of the helical structure.  In the generator $g_i$, the parameter $\theta_i$ is the rotation term about $\bfe$ and $\tau_i$ is the translation along $\bfe$. In this parameterization, $g_i$ maps the starting atom $\bfx$ to its nearest neighbors $g_1(\bfx)$ and $g_2(\bfx)$, as shown in Fig.~\ref{fig:helical_structure}(a).
To avoid degeneracies (lines, points, rings) and accumulation points, the parameters in $g_{1}$ and $g_{2}$ need to satisfy the following non-degeneracy and discreteness conditions:
\begin{align}
& \mathbf{x} \neq \mathbf{z}+\lambda \mathbf{e}, \quad \lambda \in \mathbb{R},  \nonumber \\
& p^{\star} \theta_{1}+q^{\star} \theta_{2}=2 \pi, \nonumber \\
& p^{\star} \tau_{1}+q^{\star} \tau_{2}=0 
\label{eq:discreteness}
\end{align}
for some $\left(p^{\star}, q^{\star}\right) \in \mathbb{Z}^{2}\setminus(0,0)$ and $\theta_{1} \tau_{2} \neq \theta_{2} \tau_{1}$. These conditions basically guarantee that the generated helical structure is closed and discrete (rotations sum to $2\pi$ and translations sum to zero) as shown in Fig.~\ref{fig:helical_structure}(b), and the structure is not a line (e.g. $\theta_1 \tau_2 = \theta_2 \tau_1$ will make $g_1^p(\bfx)$ and $g_2^q(\bfx)$ on a single helix).
Then, a helical structure is the collection of atomic points given by
\begin{equation}
\calH(\mathbf{x})=\left\{\mathbf{y}_{\mathbf{x}}(p, q)\left|\mathbf{y}_{\mathbf{x}}(p, q):=\mathbf{R}_{p \theta_{1}+q \theta_{2}}(\mathbf{x}-\mathbf{z})+\left(p \tau_{1}+q \tau_{2}\right) \mathbf{e}+\mathbf{z}, \quad p \in \mathbb{Z}, q=1, \ldots,\right| q^{\star} \mid\right\} 
\end{equation}
\begin{figure}
    \centering
    \includegraphics[width=0.7\linewidth]{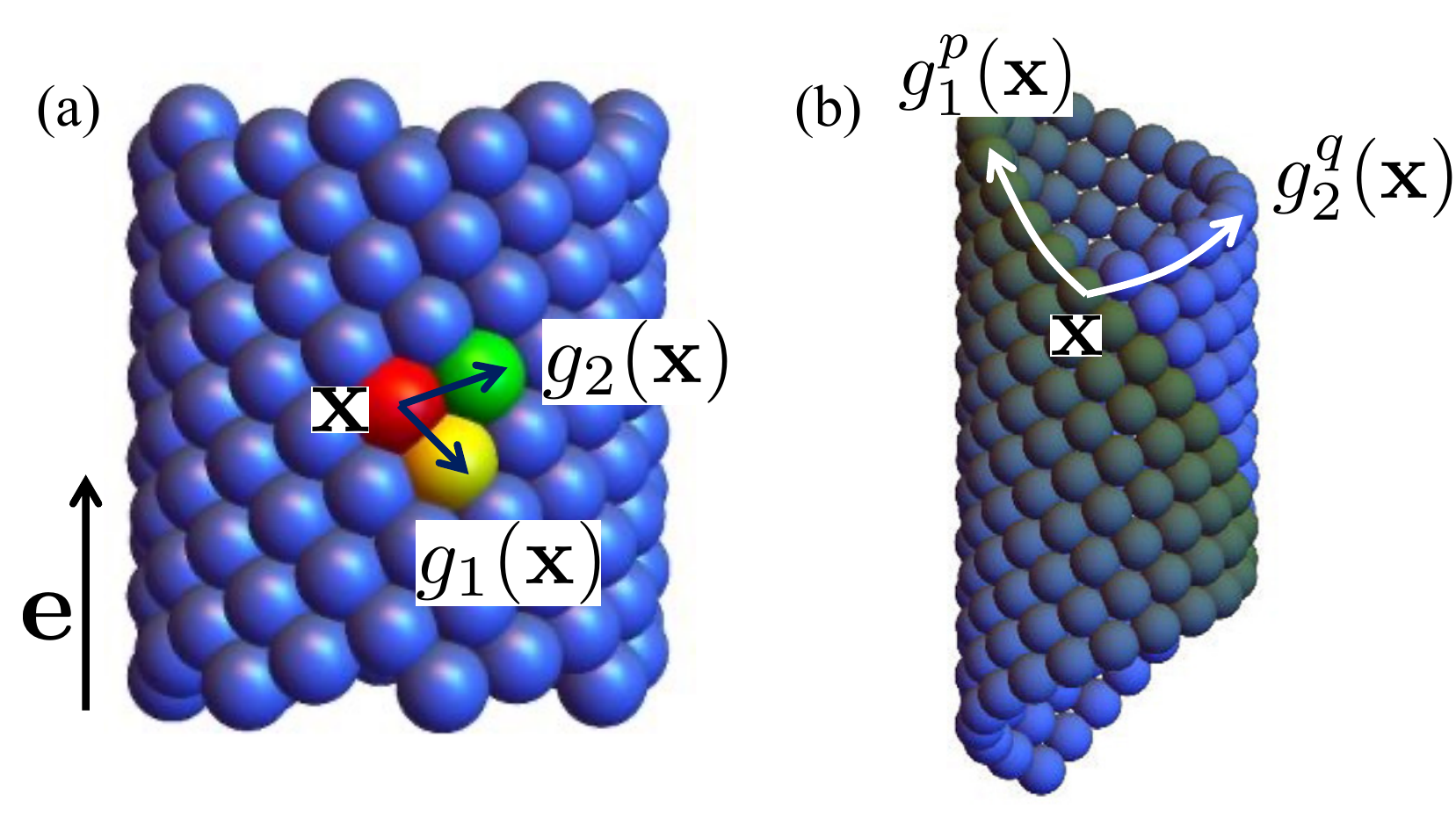}
    \caption{(a) Notations for helical structures. (b) Non-degeneracy and discreteness conditions.}
    \label{fig:helical_structure}
\end{figure}

We construct an isometric function mapping a helical structure $\calH(\mathbf{x})$ to a corresponding Bravais lattice as follows:
\begin{equation}
d(\calH(\mathbf{x})):=r \tilde{\mathbf{e}}_{1}+((\mathbf{x}-\mathbf{z}) \cdot \mathbf{e}) \mathbf{e}+\left(p \tau_{1}+q \tau_{2}\right) \mathbf{e}+\left(p r \theta_{1}+q r \theta_{2}\right) \tilde{\mathbf{e}}_{1}, \quad p \in \mathbb{Z}, q=1, \ldots,\left|q^{\star}\right|
\end{equation}
where $\tilde{\mathbf{e}}_{1}=(\mathbf{I}-\mathbf{e} \otimes \mathbf{e})(\mathbf{x}-\mathbf{z})/|(\mathbf{I}-\mathbf{e} \otimes \mathbf{e})(\mathbf{x}-\mathbf{z})|$ and $r=|(\bfI - \bfe \otimes \bfe)(\bfx - \bfz)|$ is the radius. The function $\tilde{f}$ is an isometric deformation that `unrolls' a helical structure into a 2D Bravais lattice. The corresponding 2D lattice vectors are therefore $\mathbf{h}_{1}=\left(r \theta_{1}, \tau_{1}\right)$ and $\mathbf{h}_{2}=\left(r \theta_{2}, \tau_{2}\right)$ in the basis $\left\{\tilde{\mathbf{e}}_{1}, \mathbf{e}\right\}$.  The periodicity of the bilayer helical structure aligns with that of the unrolled bilayer 2D Bravais lattices.
However, the helical structure also requires additional constraints--the discreteness condition (\ref{eq:discreteness}). Therefore, the strategy for studying the periodicity of bilayer helical structures (helical moiré pattern) is to analyze the corresponding bilayer 2D Bravais lattices with these added discreteness constraints.

\subsubsection{Helical moiré  pattern}
Before analyzing the periodicity of helical moiré  pattern, we first derive a useful theorem regarding the general property of reciprocal vectors.

\begin{theorem}\label{thm}
Suppose $p_{a}^{\star} \mathbf{r}_{1}+q_{a}^{\star} \mathbf{r}_{2}=(\lambda, 0)$ and $p_{b}^{\star} \mathbf{s}_{1}+q_{b}^{\star} \mathbf{s}_{2}=(\lambda, 0)$ for some $\left(p_{a}^{\star}, q_{a}^{\star}\right) \neq\left(p_{b}^{\star}, q_{b}^{\star}\right) \in$ $\mathbb{Z}^{2} \backslash\{\mathbf{0}\}, \lambda \in \mathbb{R} \backslash\{0\}, \mathbf{r}_{1,2}, \mathbf{s}_{1,2} \in \mathbb{R}^{2}$ subject to nondegenerate conditions $\mathbf{r}_{1} \cdot \mathbf{r}_{2}^{\perp} \neq 0, \mathbf{s}_{1} \cdot \mathbf{s}_{2}^{\perp} \neq 0$. Let $\left(\hat{\mathbf{r}}_{1}, \hat{\mathbf{r}}_{2}\right)$, $\left(\hat{\mathbf{s}}_{1}, \hat{\mathbf{s}}_{2}\right)$ be the corresponding reciprocal vectors of $\left(\mathbf{r}_{1}, \mathbf{r}_{2}\right)$ and $\left(\mathbf{s}_{1}, \mathbf{s}_{2}\right)$, respectively. Let $\hat{\mathbf{t}}_{1}=\hat{\mathbf{s}}_{1}-\hat{\mathbf{r}}_{1}$ and $\hat{\mathbf{t}}_{2}=\hat{\mathbf{s}}_{2}-\hat{\mathbf{r}}_{2}$ be the difference of reciprocal vectors.  If $\left(\mathbf{t}_{1}, \mathbf{t}_{2}\right)$ is the reciprocal vectors with respect to $\left(\hat{\mathbf{t}}_{1}, \hat{\mathbf{t}}_{2}\right)$ under $\hat{\mathbf{t}}_{1} \cdot \hat{\mathbf{t}}_{2}^{\perp} \neq 0$, then
\begin{equation}
\left(p_{b}^{\star}-p_{a}^{\star}\right) \mathbf{t}_{1}+\left(q_{b}^{\star}-q_{a}^{\star}\right) \mathbf{t}_{2}=(\lambda, 0). \label{eq:thm1}
\end{equation}
\end{theorem}

\begin{proof}
For two general vectors $\mathbf{u}_{1}, \mathbf{u}_{2} \in \mathbb{R}^{2},\left|\mathbf{u}_{1}\right| \neq 0,\left|\mathbf{u}_{2}\right| \neq 0$, using the standard procedure, the corresponding reciprocal vectors are
\begin{align}
& \hat{\mathbf{u}}_{1}=\frac{\left|\mathbf{u}_{2}\right|^{2}}{\left|\mathbf{u}_{1}\right|^{2}\left|\mathbf{u}_{2}\right|^{2}-\left(\mathbf{u}_{1} \cdot \mathbf{u}_{2}\right)^{2}} \mathbf{u}_{1}-\frac{\mathbf{u}_{1} \cdot \mathbf{u}_{2}}{\left|\mathbf{u}_{1}\right|^{2}\left|\mathbf{u}_{2}\right|^{2}-\left(\mathbf{u}_{1} \cdot \mathbf{u}_{2}\right)^{2}} \mathbf{u}_{2}, \nonumber \\
& \hat{\mathbf{u}}_{2}=-\frac{\left|\mathbf{u}_{1}\right|^{2}}{\left|\mathbf{u}_{1}\right|^{2}\left|\mathbf{u}_{2}\right|^{2}-\left(\mathbf{u}_{1} \cdot \mathbf{u}_{2}\right)^{2}} \mathbf{u}_{1}+\frac{\mathbf{u}_{2}}{\left|\mathbf{u}_{1}\right|^{2}\left|\mathbf{u}_{2}\right|^{2}-\left(\mathbf{u}_{1} \cdot \mathbf{u}_{2}\right)^{2}} \mathbf{u}_{2}. \label{eq:reci}
\end{align}
Since $\left(p_{i}^{\star}, q_{i}^{\star}\right) \neq(0,0), i \in\{a, b\}$, we first assume $q_{a}^{\star} \neq 0$ and substitute $\mathbf{r}_{2}=1 / q_{a}^{\star}\left((\lambda, 0)-p_{a}^{\star} \mathbf{r}_{1}\right)$ for the calculation of reciprocal vectors using (\ref{eq:reci}). Similarly, we substitute $\mathbf{s}_{2}=1 / q_{b}^{\star}\left((\lambda, 0)-p_{b}^{\star} \mathbf{s}_{1}\right)$ for the calculation of $\hat{\mathbf{s}}_{i}$ if $q_{b}^{\star} \neq 0$. Subsequently, we can compute $\hat{\mathbf{t}}_{1}, \hat{\mathbf{t}}_{2}$ as functions of $\mathbf{r}_{i}$ and $\mathbf{s}_{j}$, where $i, j$ is 1 or 2 depending on the nonzero $q^\star_{a,b}$. Then, we use (\ref{eq:reci}) again to calculate $\mathbf{t}_{1}$ and $\mathbf{t}_{2}$, and confirm that (\ref{eq:thm1}) holds by direct calculation.
\end{proof}

Now suppose we have two distinct helical structures given by
\beqs
\calH_{a}(\mathbf{x}) & =&\left\{\mathbf{y}_{\mathbf{x}}^{a}(p, q)\left|\mathbf{y}_{\mathbf{x}}^{a}(p, q):=\mathbf{R}_{p \theta_{1 a}+q \theta_{2 a}}(\mathbf{x}-\mathbf{z})+\left(p \tau_{1 a}+q \tau_{2 a}\right) \mathbf{e}+\mathbf{z}, p \in \mathbb{Z}, q=1, \ldots,\right| q_{a}^{\star} \mid\right\} \\
\calH_{b}(\mathbf{x}) & =&\left\{\mathbf{y}_{\mathbf{x}}^{b}(p, q)\left|\mathbf{y}_{\mathbf{x}}^{b}(p, q):=\mathbf{R}_{p \theta_{1 b}+q \theta_{2 b}}(\mathbf{x}-\mathbf{z})+\left(p \tau_{1 b}+q \tau_{2 b}\right) \mathbf{e}+\mathbf{z}, p \in \mathbb{Z}, q=1, \ldots,\right| q_{b}^{\star} \mid\right\}
\eeqs
where the parameters in $a$ and $b$ satisfy the discreteness and non-degeneracy conditions
\[
\left\{\begin{array}{l}
p_{a}^{\star} \theta_{1 a}+q_{a}^{\star} \theta_{2 a}=2 \pi  \tag{15}\\
p_{a}^{\star} \tau_{1 a}+q_{a}^{\star} \tau_{2 a}=0

\end{array}\right.
\]
and
\[
\left\{\begin{array}{l}
p_{b}^{\star} \theta_{1 b}+q_{b}^{\star} \theta_{2 b}=2 \pi  \tag{16}\\
p_{b}^{\star} \tau_{1 b}+q_{b}^{\star} \tau_{2 b}=0
\end{array}\right.
\]
for some $\left(p_{a}^{\star}, q_{a}^{\star}\right),\left(p_{b}^{\star}, q_{b}^{\star}\right) \in \mathbb{Z}^{2}\setminus(0,0), \theta_{1 a} \tau_{2 a} \neq \theta_{2 a} \tau_{1 a}, \theta_{1 b} \tau_{2 b} \neq \theta_{2 b} \tau_{1 b}$. The corresponding 2D Bravais lattice vectors are
\begin{align}
& \mathbf{r}_{1 a}=\left(r \theta_{1 a}, \tau_{1 a}\right), \quad \mathbf{r}_{2 a}=\left(r \theta_{2 a}, \tau_{2 a}\right),\nonumber \\
& \mathbf{r}_{1 b}=\left(r \theta_{1 b}, \tau_{1 b}\right), \quad \mathbf{r}_{2 b}=\left(r \theta_{2 b}, \tau_{2 b}\right).
\end{align}
Following the discreteness conditions, the lattice vectors satisfy
\begin{align}
p_{a}^{\star} \mathbf{r}_{1 a}+q_{a}^{\star} \mathbf{r}_{2 a} & =(2 \pi r, 0), \nonumber \\
p_{b}^{\star} \mathbf{r}_{1 b}+q_{b}^{\star} \mathbf{r}_{2 b} & =(2 \pi r, 0).
\end{align}

Let $\left(\hat{\mathbf{r}}_{1 a}, \hat{\mathbf{r}}_{2 a}\right)$ be the reciprocal vectors of $\left(\mathbf{r}_{1 a}, \mathbf{r}_{2 a}\right)$, and $\left(\hat{\mathbf{r}}_{1 b}, \hat{\mathbf{r}}_{2 b}\right)$ be the reciprocal vectors of $\left(\mathbf{r}_{1 b}, \mathbf{r}_{2 b}\right)$. The differences of reciprocal vectors, $\left(\hat{\mathbf{m}}_{1}, \hat{\mathbf{m}}_{2}\right)=\left(\hat{\mathbf{r}}_{1 b}-\hat{\mathbf{r}}_{1 a}, \hat{\mathbf{r}}_{2 b}-\hat{\mathbf{r}}_{2 a}\right)$, are the reciprocal vectors of the moiré vectors $\left(\mathbf{m}_{1}, \mathbf{m}_{2}\right)$, as discussed in Section \ref{sec:2D}. Recalling Theorem \ref{thm}, we know that $\left(\mathbf{m}_{1}, \mathbf{m}_{2}\right)$ satisfies the similar discreteness condition:
\begin{equation}
\left(p_{b}^{\star}-p_{a}^{\star}\right) \mathbf{m}_{1}+\left(q_{b}^{\star}-q_{a}^{\star}\right) \mathbf{m}_{2}=(2 \pi r, 0) 
\end{equation}

This equation implies that the formed moiré pattern is also a discrete helical structure. We may write $\mathbf{m}_{1}$ and $\mathbf{m}_{2}$ as
\begin{equation}
\mathbf{m}_{1}=\left(r \theta_{1 m}, \tau_{1}\right), \quad \mathbf{m}_{2}=\left(r \theta_{2 m}, \tau_{2}\right).
\end{equation}
Then the helical moiré pattern described by
\begin{equation}
\mathbf{y}_{m}(p, q)=\mathbf{R}_{p \theta_{1 m}+q \theta_{2 m}}(\mathbf{x}-\mathbf{z})+\left(p \tau_{1 m}+q \tau_{2 m}\right) \mathbf{e}+\mathbf{z} 
\end{equation}
is a  helical structure if $\left(p_{b}^{\star}-p_{a}^{\star}, q_{b}^{\star}-q_{a}^{\star}\right) \neq \mathbf{0}$ and $\theta_{1 m} \tau_{2 m} \neq \theta_{2 m} \tau_{1 m}$. Now we compute the moiré vectors $\left(\mathbf{m}_{1}, \mathbf{m}_{2}\right)$ explicitly. The procedure is straightforward:
\begin{enumerate}
  \item Given $\left(\mathbf{r}_{1 a}, \mathbf{r}_{2 a}\right)$ and $\left(\mathbf{r}_{1 b}, \mathbf{r}_{2 b}\right)$, compute the corresponding reciprocal vectors $\left(\hat{\mathbf{r}}_{1 a}, \hat{\mathbf{r}}_{2 a}\right)$ and $\left(\hat{\mathbf{r}}_{1 b}, \hat{\mathbf{r}}_{2 b}\right)$.

  \item Compute $\left(\hat{\mathbf{m}}_{1}, \hat{\mathbf{m}}_{2}\right)=\left(\hat{\mathbf{r}}_{1 b}-\hat{\mathbf{r}}_{1 a}, \hat{\mathbf{r}}_{2 b}-\hat{\mathbf{r}}_{2 a}\right)$.

  \item Compute the moiré vectors $\left(\mathbf{m}_{1}, \mathbf{m}_{2}\right)$ as the reciprocal vectors of ( $\hat{\mathbf{m}}_{1}, \hat{\mathbf{m}}_{2}$ ) using (\ref{eq:reci}).

\end{enumerate}

By direct calculation, $\mathbf{m}_{1}$ and $\mathbf{m}_{2}$ have the analytical forms as:
\beqs
\mathbf{m}_{1}=\left(r \frac{\theta_{1 b}(\theta_{2 a} \tau_{1 a}-\theta_{1 a} \tau_{2 a})+\theta_{1 a}(\theta_{1 b} \tau_{2 b}-\theta_{2 b} \tau_{1 b})}{(\theta_{2 a}-\theta_{2 b})(\tau_{1 a}-\tau_{1 b})+(\theta_{1 b}-\theta_{1 a})(\tau_{2 a}-\tau_{2 b})}, \frac{\tau_{1 b}(\theta_{2 a} \tau_{1 a}-\theta_{1 a} \tau_{2 a})+\tau_{1 a}(\theta_{1 b} \tau_{2 b}-\theta_{2 b} \tau_{1 b})}{(\theta_{2 a}-\theta_{1 b})(\tau_{1 a}-\tau_{1 b})+(\theta_{1 b}-\theta_{1 a})(\tau_{2 a}-\tau_{2 b})}\right), \nonumber \\
\mathbf{m}_{2}=\left(r \frac{\theta_{2 b}(\theta_{2 a} \tau_{1 a}-\theta_{1 a} \tau_{2 a})+\theta_{2 a}(\theta_{1 b} \tau_{2 b}-\theta_{2 b} \tau_{1 b})}{(\theta_{2 a}-\theta_{2 b})(\tau_{1 a}-\tau_{1 b})+(\theta_{1 b}-\theta_{1 a})(\tau_{2 a}-\tau_{2 b})}, \frac{\tau_{2 b}(\theta_{2 a} \tau_{1 a}-\theta_{1 a} \tau_{2 a})+\tau_{2 a}(\theta_{1 b} \tau_{2 b}-\theta_{2 b} \tau_{1 b})}{(\theta_{2 a} -\theta_{2b})(\tau_{1 a}-\tau_{1 b})+(\theta_{1 b}-\theta_{1 a})(\tau_{2 a}-\tau_{2 b})}\right). \label{eq:m1m2}
\eeqs
We write the above equations in condensed forms as
$\bfm_1=(r\theta_{1m}, \tau_{1m})$ and $
\bfm_2=(r\theta_{2m}, \tau_{2 m})$. 
Finally, we can define the generators of the moiré pattern as
\begin{equation}
g_{1 m}=\left(\mathbf{R}_{\theta_{1 m}} \mid\left(\mathbf{I}-\mathbf{R}_{\theta_{1 m}}\right) \mathbf{z}+\tau_{1 m} \mathbf{e}\right), \quad g_{2 m}=\left(\mathbf{R}_{\theta_{2 m}} \mid\left(\mathbf{I}-\mathbf{R}_{\theta_{2 m}}\right) \mathbf{z}+\tau_{2 m} \mathbf{e}\right). 
\end{equation}

\subsubsection{Example}
Here we provide an example in Fig.~\ref{fig:helical} to illustrate the periodicity of the bilayer helical structures and the helical moiré pattern. The parameters for $\calH_{a}(\mathbf{x})$ and $\calH_{b}(\mathbf{x})$ are:

\begin{enumerate}
  \item $r=1,(\theta_{1a},\tau_{1a})=(2\pi/175,-0.048), (\theta_{2a},\tau_{2a})=(4\pi/175,0.024),(p_{a}^{\star}, q_{a}^{\star})=(35,70)$.
  \item $r=1,(\theta_{1b},\tau_{1b})=(2\pi/250,-0.024), (\theta_{2b},\tau_{2b})=(2\pi/75,0.02),\left(p_{b}^{\star}, q_{b}^{\star}\right)=(50,60)$.
\end{enumerate}
According to Eq.~(\ref{eq:m1m2}), the  parameters for the formed helical moiré pattern are:
$$
r=1, (\bar{\theta}_{1},\bar{\tau}_{1})=(0.074,0.071), (\bar{\theta}_{2},\bar{\tau}_{2})=(0.739,0.106), (\bar{p}^\star,\bar{q}^\star) = (15, -10).
$$

\begin{figure}[!h]
    \centering
    \includegraphics[width=0.7\textwidth]{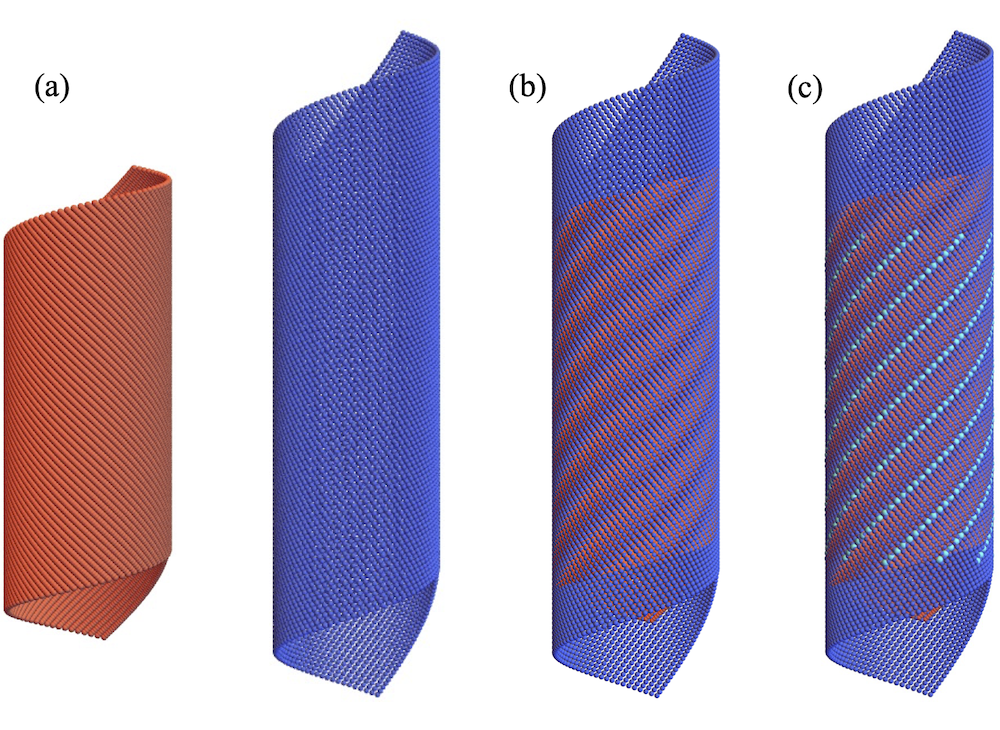}
    \caption{(a) Two helical structures with parameters $(\theta_{1a},\tau_{1a})=(2\pi/175,-0.048), (\theta_{2a},\tau_{2a})=(4\pi/175,0.024)$ and $(\theta_{1b},\tau_{1b})=(2\pi/250,-0.024), (\theta_{2b},\tau_{2b})=(2\pi/75,0.02)$. (b) The superposed moiré pattern. (c) The moiré pattern is indicated by light blue atoms with the calculated parameters $(\bar{\theta}_{1},\bar{\tau}_{1})=(0.074,0.071), (\bar{\theta}_{2},\bar{\tau}_{2})=(0.739,0.106)$.}
    \label{fig:helical}
\end{figure}

\section{Conformal moiré pattern}
A group $\calG$ is a conformal Euclidean group if $\calG$ is a group of affine linear maps of the form $(\eta \bfQ|\bfc)$ with $\eta>0, \bfQ \in \text{O}(3), \bfc \in \mathbb{R}^3$, and some elements have $\eta \neq 1$. 
An example of conformal Euclidean groups \footnote{Arun Soor derived a full characterization in a note during his visit in Richard James' group.} is given by
\beq
\calG = \{g^p h^q: (p,q) \in \mathbb{Z}^2\},~ g=(\eta \bfR_{\theta}|(\bfI - \eta\bfR_{\theta}) \bfz),~ h = ( \bfR_{\varphi}|(\bfI - \bfR_{\varphi}) \bfz)
\eeq
where $\eta \in \mathbb{R}$, $\eta \neq 1$, $\bfR_{\theta}, \bfR_{\varphi} \in \mathrm{SO}(3)$, $\varphi = 2\pi/n, n\in \mathbb{Z}$. 
Then the conformal structure generated by $\calG$ is given by
\beq
\calG(\bfx) = \{\eta^p \bfR_{p\theta+q\varphi} (\bfx - \bfz) + \bfz:(p,q) \in \mathbb{Z}^2\}.
\eeq
The parameter $\eta \neq 1$ yields a dilatation in the direction perpendicular to the axis of $\bfR_{\theta}$, and therefore, a conformal structure is not an objective structure. 
It should be noted that conformal structures are not physically realistic, as the dilation parameter  $\eta$  can lead to accumulation points with divergent energy or infinitely distant points that occupy an unphysically large space. A strategy to avoid this is to cut off the domain of $(p,q)$.
Given the variety of conformal structures, it is useful to analyze the corresponding moiré patterns locally, following the group orbit framework and the Fourier approach.
\begin{figure}[!h]
    \centering
    \includegraphics[width=\textwidth]{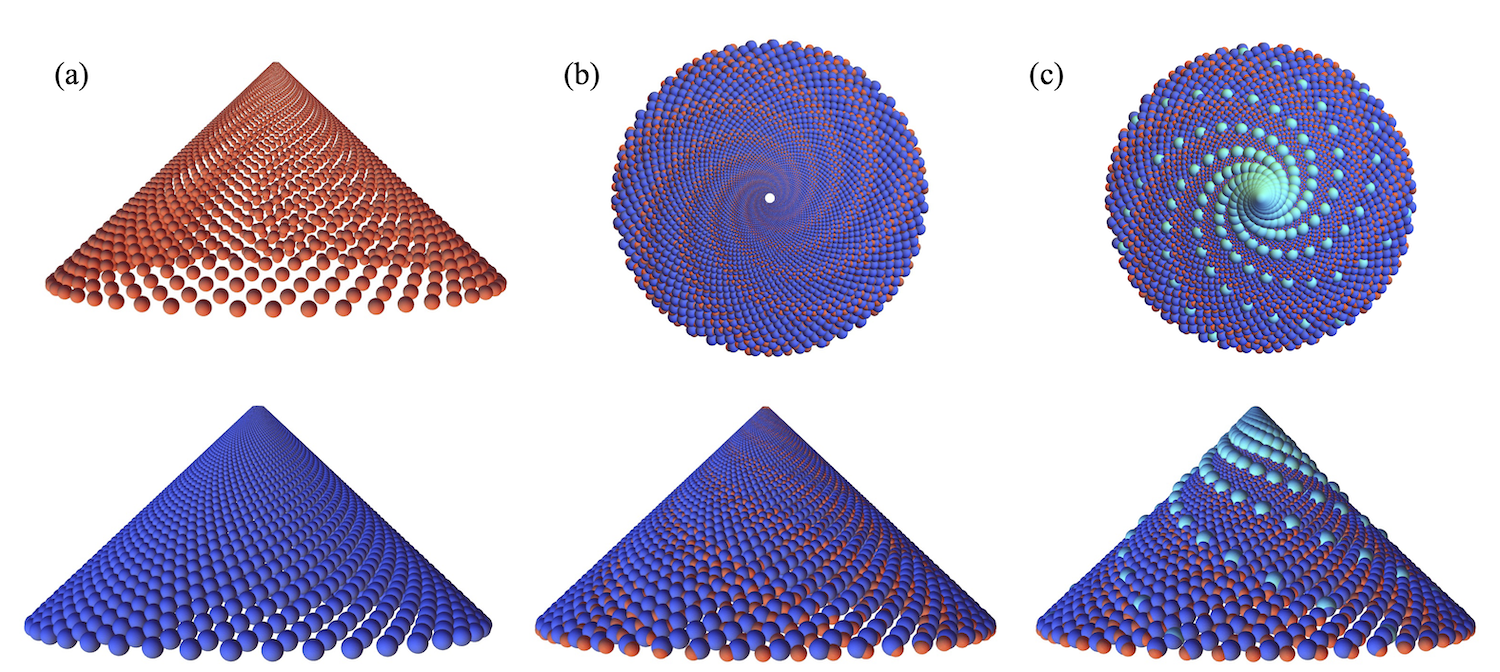}
    \caption{(a) Two conformal structures with parameters $\bfz=(0,0,0),\bfx=(0.1,0,0.1), \eta_a=1.03, \theta_a=2\pi/31, \varphi_a=-2\pi/43$ and $\eta_b=1.04, \theta_b=2\pi/35, \varphi_b=-2\pi/50$. (b) Two views of the superposed moiré pattern. (c) The conformal moiré pattern is indicated by light blue atoms with parameters $\bar{\eta}=0.887, \bar{\theta}= 0.36\pi, \bar{\varphi}=-2\pi/7$.}
    \label{fig:conformal}
\end{figure}

Now we superpose two conformal structures with distinct parameters $\eta, \theta$ and $\varphi$. In the example shown in Fig.~\ref{fig:conformal}, the blue and red structures have parameters $\bfx=(0.1,0,0.1), \bfz=(0,0,0),\eta_a =1.03, \theta_a=2\pi/31, \varphi_a=-2\pi/43$ and $\eta_b=1.04, \theta_b=2\pi/35, \varphi_b=-2\pi/50$ respectively, ensuring that the two structures are distinct and close. Again, we use an isometric deformation that maps the cone-like structures onto a plane and then compute the moiré pattern with the Fourier approach. Specifically, the isometric deformation is chosen as (in polar coordinates $(r,\theta)$)
\beq
f(\calG(\bfx)) = (\eta^p|\bfx - \bfz|, p \theta + q \varphi). \label{eq:polar}
\eeq
We then derive the moiré pattern for the 2D lattices given by the primitive lattice vectors 
\beqs
&&\bfr_1 = (\log\eta_a, \theta_a),~ \bfr_2=(0,\varphi_a), \nonumber\\
&&\bfs_1 = (\log\eta_b, \theta_b), ~ \bfs_2=(0,\varphi_b).
\eeqs
Notice that here we revise the $r$ coordinate in (\ref{eq:polar}) to form a `2D lattice'.
Following the Fourier approach, the lattice vectors for the moiré pattern are given by 
\beqs
(\log\bar{\eta}, \bar{\theta}) &=& \left(\frac{\log \eta_a \log \eta_b}{\log \eta_a - \log \eta_b}, \frac{\theta_b \varphi_a \log \eta_a- \theta_a \varphi_b \log\eta_b}{(\varphi_a-\varphi_b)(\log\eta_a - \log\eta_b)}\right),\nonumber\\
(0, \bar{\varphi})&=&\left(0, \frac{\varphi_a\varphi_b}{\varphi_b - \varphi_a}\right),
\eeqs
yielding the parameters  $\bar{\eta}=0.887, \bar{\theta}= 0.36\pi, \bar{\varphi}=-2\pi/7$ for the moiré pattern in light blue shown in Fig.~\ref{fig:conformal}(c). 

For the objective and conformal moiré patterns discussed in this paper, we use a similar Fourier approach to study their geometry, with an additional modification needed for the conformal case. Our analysis primarily focuses on a 2D plane, as the lattice points in these examples lie on developable surfaces—planes, cylinders, and cones—that are isometric to a plane. For other types of objective structures, such as those on a sphere, further consideration is needed.

\section{Discussion}
Objective structure, as a concept of general periodic atomic structures, has been applied in the study of phase transformations, molecular dynamics, origami design, etc. In this paper, we contribute to a new theme about using objective structures to study moiré patterns. By superposing two layers of objective structures, objective moiré patterns are constructed. We mainly showcase three objective moiré patterns formed by ring structures, 2D lattice structures and helical structures, respectively. We provide the analytical formulas for the geometries of these moiré patterns. Not surprisingly, the moiré patterns follow the symmetry of the original objective structures, but of course, have different lattice parameters. In addition, we give an example of a conformal moiré pattern constructed by overlaying two conformal structures (which are not objective structures), indicating the versatility of the group orbit idea.

The method we use to compute the geometry of objective moiré patterns is the augmented Fourier approach inspired by the classic 2D lattice case. Notice that the three examples of objective moiré patterns we solve here can be mapped isometrically onto a plane. Thus we may solve the 2D moiré pattern on a plane first, then map it back to the objective structures to obtain the parameters. For the conformal structure case, it is also possible to map it onto a plane, but it needs a further transformation to obtain sensible 2D lattices. However, not all objective structures can be mapped isometrically onto a plane, for example if the effective surface of the lattice points has non-zero Gaussian curvature (i.e. buckyballs). In that sense, the problem of objective moiré pattern is not fully solved yet. Recalling the general formulation of objective structures $\{g_1(\bfx), g_2(\bfx),\dots\}$ and the corresponding isometry group $\{g_1,g_2,\dots\}$, we may propose a possible approach of solving the problem directly in the space of objective structures, not restricted on a plane. The key point is to find a generalized ``wave function'' that reaches the maximal magnitude at lattice points. Then the periodicity/geometry of the moiré pattern is simply the ``wave packet'' by superposing the two sets of wave functions. This has been done partially for the helical structures (the wave functions have been found) \cite{banerjee2015spectral, justel2016bragg}, but the general approach is still to be developed.

This work is purely geometric, aiming to offer foundational insights into more complex moiré patterns, especially those in 3D, which remain largely underexplored. In future studies, it would be valuable to integrate mechanics and physics into the study of objective moiré patterns, including aspects such as membrane deformability, X-ray diffraction, wave manipulation, molecular dynamics, and other applications.

\section*{Acknowledgements}
This paper is dedicated to Dick James on the occasion of his 70$^\text{th}$ birthday. Dick is a great researcher, mentor, colleague and friend to many of us. The financial support for this work was provided by the National Natural Science Foundation of China (grant no. 12472061).
\bibliographystyle{abbrv}
\bibliography{bilayer.bib} 
\end{document}